\documentclass[12pt,a4paper]{article}
\usepackage{newunicodechar}
\newunicodechar{ﬁ}{fi}
\newunicodechar{ﬀ}{ff}

\usepackage{amsmath}
\usepackage{amsfonts}
\usepackage{amssymb}
\usepackage{natbib}
\usepackage{algorithm}
\usepackage{algpseudocode}
\usepackage{amsthm}
\usepackage{graphicx}
\usepackage{mathtools}
\DeclarePairedDelimiter{\norm}{\lVert}{\rVert} 
\usepackage[utf8]{inputenc}

\numberwithin{equation}{section}

\usepackage{float}

\begin{document}

\begin{titlepage}
   \vspace*{\stretch{1.0}}
   \begin{center}
      \Large\textbf{Some connections between higher moments portfolio optimization methods}\\
	  \large\textit{Farshad Noravesh}\footnote{Email: noraveshfarshad@gmail.com}	  
	  \large\textit{Kristiaan Kerstens} \footnote{Email: k.kerstens@ieseg.fr}      
      
   \end{center}
   \vspace*{\stretch{2.0}}
\begin{abstract}
In this paper, different approaches to portfolio optimization having higher moments such as skewness and kurtosis are classified so that the reader can observe different paradigms and approaches in this field of research which is essential for practitioners in Hedge Funds in particular. Several methods based on different paradigms such as utility approach and multi-objective optimization are reviewed and the advantage and disadvantageous of these ideas are explained.
Keywords: multi-objective optimization, portfolio optimization, scalarization, utility
\end{abstract}

\end{titlepage}

\section{Introduction}

One can see the huge articles on higher moments of portfolio optimization as a spectrum of assumptions which varies from the most general and weakest assumptions to the most restricted and strong assumptions. Classification of these articles can be done in different contexts. An important context is whether we use parametric modeling or nonparametric modeling. It is interesting that these two paradigms have also been discussed in other fields of research such as machine learning where kernel methods such as kernel support vector machine(kernel-SVM) or nonlinear kernel dimensionality reduction is a nonparametric approach while methods like variational auto-encoder(VAE) is an example of parametric approaches. These two different paradigms has its root in statistics where we have parametric and nonparametric approach to solve statistical problems. An abstract way of explaining these two paradigms is well understood of we consider parametric approach as a way to make a finite dimensional approximation of probability distributions of returns, while in the nonparametric approach, no explicit form of distribution is assumed and our ignorance on reality avoids us to make strong assumptions of reality. In reality, space of distributions on asset returns are extremely complex. Thus, using parametric distributions builds a finite dimensional structure on the space of return distributions and therefore on the moments of the returns which are used in the objective function of portfolio optimization. The importance of utility function is well understood specially if a taylor series expansion of utility around the expected return is written to realize the connection of utility with higher moments such as skewness and kurtosis. Thus,
\begin{equation}
u(R)=u(E(R))+u'[E(R)][R-E(R)]+\frac{ u"[E(R)]}{{2!}}  \sigma^2_R+\frac{u'''[E[R]]}{3!} M^3_R
\end{equation}        
where $\sigma^2_R$   is the variance of expected return, and is the third moment about the mean of expected return which is skewness. The following utility functions are used in the literature:
\begin{equation}
\begin{split}
u(R)&=ln(\lambda R)  \\
u(R)&=\sqrt{\lambda R}  \\
u(R)&=-e^{-\lambda R}
\end{split}
\end{equation}
However, the third utility function is used more relatively for its nice properties.  Taking the expectation of utility after Taylor expansion would produce:
\begin{equation}
\begin{split}
E(u(R))&=-e^{\lambda E(R)}(1+\frac{\lambda}{1!}E(R-E(R))-\frac{\lambda^2}{2!}(R-E(R))^2 +\frac{\lambda^3}{31!}E(R-E(R))^3  \\
&-\frac{\lambda^4}{4!}(R-E(R))^4+O(R^5)  )  
\end{split}
\end{equation}
Thus, the following optimization problem is produced:
 \begin{equation}
min_{\omega} -\omega'\mu+\lambda_1 \omega' \Sigma \omega-\lambda_2\omega'M_3 (\omega \otimes \omega) +\lambda_3 \omega' M_4(\omega \otimes \omega \otimes \omega)
\end{equation}
where   $\omega \otimes \omega$   stands for Kronecker product and, $\lambda_1=\frac{\lambda^2}{2!} $   , $\lambda_2=\frac{\lambda^3}{3!}$ , $\lambda_3=\frac{\lambda^4}{4!}$ 
\citep{Glawischnig2013} has used an iterative approach to solve this objective function without using polynomial goal programming(PGP). They start with $\lambda$ equal to twenty and then reduce it in each step until two to avoid short selling of more than one hundred percent. The algorithm is multi-start to get robust results. The trick that \citep{Glawischnig2013} has used is that in the first iteration the skewness and kurtosis is calculated and then fixed to make the full optimization problem a valid quadratic optimization to be solved by quadratic solvers, but this iterative approach that changes the nature of the optimization problem and neglects the third order polynomials and 4th order polynomials could result a big difference. \citep{Glawischnig2013},\citep{Jondeau2006} almost use similar ideas and their approach has an advantage over PGP since the weights of the optimization are naturally formed by the taylor approximation of utility rather than arbitrary selection that is used in PGP. 
\section{Parametric and nonparametric modeling} 
The motivation comes from many different intuitions such as constraining the structure of return distributions of financial securities and approximate the complex nature of asset returns. \citep{Adcock2014} uses coherent multivariate probability distribution assumption for asset returns to be able to use stein’s lemma to create the mean-variance-skewness efficient hyper-surface. So the assumption is that the asset returns are defined by the convolution of a multivariate elliptically symmetric distribution and a multivariate distribution of non-negative random variables so that the efficient portfolios could be computed using quadratic programming on the efficient surface. Many researchers are amazed by the Multivariate skew normal (MSN) which was first introduced in \citep{azzalini1996} and was applied in finance in \citep{Adcock1999} is the center of many portfolio selection methods. One of the motivations of using MSN is the simplicity of the maximization of utility, specially if utility is an exponential function of return as explained in \citep{Adcock2005},\citep{Landsman2019}, the drawback of such an approach is explained in \citep{Adcock2005} when the preference parameter is too high or very small. But the bigger disadvantage of such an approach is simply shaping the mean-variance-skewness efficient surface by choosing an exponential utility function. Another drawback is that the generalization of this approach for kurtosis and higher moments is not straightforward. A generalization of MSN could be seen in \citep{SAHU2003} where an analytic forms of densities are obtained and there distributional properties are studied. These parametric modeling approaches to model distributions for returns are very diverse and many researchers have noticed that such as \citep{Jondeaua2003} that uses a generalized Student-t distribution or \citep{Mencia2009} which uses location scale mixture of normals and using maximum likelihood to infer the parameters. The log-normal distribution has also been used in the literature to model asset returns, but its skewness is a function of the mean and the variance, not a separate skewness parameter. An interesting example which combines parametric modeling(log normal) of portfolio distribution with goal programming is \citep{Changetal2008}. 
\citep{Glawischnig2013},\cite{Jondeau2006} are examples of nonparametric approach to portfolio optimization in higher moments since no distribution is assumed. The problem of skewness term in the optimization is that it makes the optimization a non-convex problem and therefore not tractable. This is the motivation of \citep{Konno1993} that without any assumption on the form of probability distribution of returns tries to approximate the third order term due to skewness by a piecewise linear approximation, however this approximation should be tested on more experimental data to see to what extent the approximation is valid. \citep{Konno1998} uses a similar approach and resolves the third order nonlinearity of skewness by representing it with a difference of two convex functions and then using branch and bound to solve the mean-variance-skewness portfolio optimization problem. A more detailed and visual derivation of \citep{Konno1998} can be seen in \citep{Konno1995},\citep{Konno2005}. The problem with ideas in \citep{Konno1993},\citep{Konno1998},\citep{Konno1995},\citep{Konno2005} is that it can not be generalized to higher moments easily and the order of complexity will be high. Another example for nonparametric higher moment portfolio optimization is based on the concept of shortage function and the geometric representation of mean-variance-skewness portfolio is illustrated in \citep{Kerstens2011}  
\section{A priori, interactive and posteriori}
The literature on higher moment portfolio optimization can be classified in a different context. In this context, it is important how the preferences are presented in the optimization process and therefore three different paradigms are discussed in the literature namely a priori, interactive and posteriori. A priori methods such as goal programming and utility function method are used when the preference of investor are known beforehand. In goal programming, first all objective optimization problems are solved regardless of other objectives. Then a final scalar objective optimization is solved that has some weights that need to be fixed. Even if all possible weights are checked, still some Pareto efficient solutions may be missing. Examples of using goal programming for higher moments portfolio optimization are \citep{Aksarayli2018},\citep{Bergh2008}. Although it is a simple algorithm, using appropriate weights is a debate. It can also produce solutions that are not Pareto efficient, and some algorithm is needed to project them back to the Pareto efficient solutions. Goal programming has many variants as explained in \citep{Aouni2014},\citep{Tamiz2013}.
1- Lexicographic
2- weighted, see for example \citep{Chang2011}
3- polynomial, see for example \citep{Chang2008},\citep{Chunhachinda1997},\citep{Davies2009},\citep{Lai1991},\citep{Mhiri2010},\citep{Proelss2014} and also \citep{Ghahtarani2013} if robustness with respect to some coefficients of the optimization is a concern as well.
4- stochastic 
5- fuzzy
Another paradigm is to use posteriori methods, where the target is finding all Pareto efficient frontier. Methods such as linear weighting method, weighted geometric mean and Normal Boundary Intersection (NBI)\citep{Audet2008}, Modified Normal Boundary Intersection (MNBI), Normal Constraint, Multiple-objective Branch-and-Bound,epsilon-constraints method and Pascoletti Serafini scalarization. Finally, there is interactive paradigm where the investor iteratively solves the optimization and gets feedback from solutions to find the Pareto optimal solutions. An example of using epsilon-constraints for portfolio optimization can be seen in \citep{Xidonas2010},\citep{Xidonasetal2010},\citep{Xidonas2011} which also includes an interactive filtering process to consider investor preferences. Methods like NBI are computationally expensive since for each iteration, an optimization problem should be solved but on the other hand has a geometrical intuition which can also be used for other methods as is explained in \citep{Kanafi2015}
\section{Pascoletti Serafini Scalarization (SP)}
SP has two parameters namely a and r, which are chosen from. The method is originally from \citep{Pascoletti1984} but is generalized in \citep{Eichfelder2008}. The novelty of their approach is  on conditions to bound the parameter a on a restricted hyperplane, but it can only be used in two dimensions and therefore not applicable to our four dimensional objective function that covers skewness and kurtosis as well.  \citep{Khorram2014} attempts to find the restricted set for parameter a but only the trivial point of zero for parameter a is considered and also the simple EP cone is considered. 
\begin{figure}
  \includegraphics[width=\linewidth]{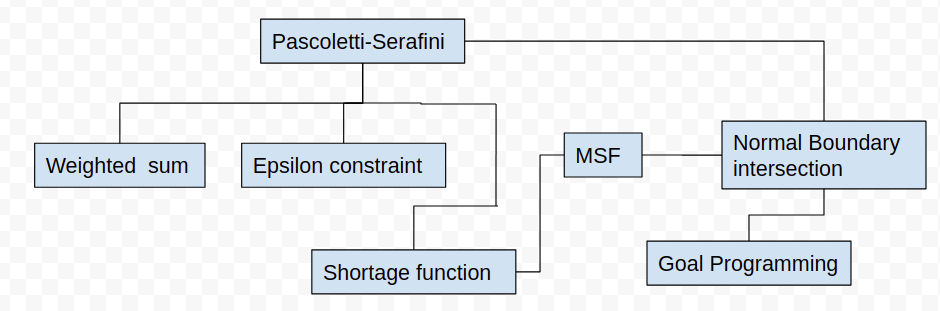}
  \caption{equivalence of multiobjective optimization algorithms}
  \label{fig:bigPicture}
\end{figure}

\subsection{Shortage Function (SF)}
\citep{Briec2004} introduced the shortage function as a portfolio performance measure in the traditional mean-variance portfolio framework.
The original shortage function is computed by solving the following problem:
\begin{equation}  \label{eq-SF}
\begin{split}
max \ \delta \\  
E[R(y^{k})]+\delta g_{E}&\leq E[R(x)] , \\
Var[R(y^{k})]-\delta g_{V} &\geq Var[R(x)] , \\
Sk[R(y^{k})]+\delta g_{S}&\leq Sk[R(x)] , \\
\sum x_i&=1 \ ,x_{i}\geq 0 , i=1,...,n
\end{split}
\end{equation}
where g is the direction vector and the effect of choosing g is well described in \citep{Kerstens2012}. A good explanation of how shortage function can be useful for mean-variance-skewness portfolio framework is explained in \citep{Briec2007}. The connection between shortage function and NBI method is explained in the present paper, however the connection between shortage function and polynomial goal programming (PGP) is well described in \citep{Briec2013}. 
Now, the following special case of shortage function is defined which is called modified shortage function (MSF) and is useful to show its connection with NBI method.
\renewcommand{\thesubsection}{\arabic{subsection}}
\newtheorem{MSFdefinition}{Definition}{subsection}
\begin{MSFdefinition}{Modified shortage function(MSF) is defined as the solution of the following optimization problem}\label{MSFdefinition}
\begin{equation}  \label{eq-MSF}
\begin{split}
max \ \delta \\  
E[R(y^{k})]+\delta g_{E}&= E[R(x)] , \\
Var[R(y^{k})]-\delta g_{V} &= Var[R(x)] , \\
Sk[R(y^{k})]+\delta g_{S}&= Sk[R(x)] , \\
\sum x_i&=1 \ ,x_{i}\geq 0 , i=1,...,n
\end{split}
\end{equation}
\end{MSFdefinition}

\subsection{NBI and SP}
In this section, the equivalence of NBI and MSF are proved. 
One of the most important scalarization technique is normal boundary intersection method (NBI) which is as follows
\begin{equation}  \label{eq-NBI}
\begin{split}
max \ s  \\
\Phi\beta+s\bar{n}&=f(x)-f^{*}   \\
s \in R ,& \ x\in\Omega
\end{split}
\end{equation}
Different optimization problems in \eqref{eq-NBI} for different  $\beta\in R^{m}_{+}$, 
having $\sum_{i=1}^{m} \beta_{i}=1$ are solved. Here $f^{*}$  denotes the so-called ideal point and the 
matrix $\Phi\in R^{m\times m}$ consists of the columns $f(x^i)-f^{*} \ ,i=1,\hdots,m$ and the vector $\bar{n}$ is defined as normal unit vector to the hyperplane directing to the negative orthant.
The problem of this method is that not all minimal points can be found as a solution of NBI. The following lemma shows the direct connection between NBI and modified SP.
\renewcommand{\thesubsection}{\arabic{subsection}}
\newtheorem{NBISPLemma}{Lemma}[subsection]
\begin{NBISPLemma}{\citep{Eichfelder2008}}\label{NBISP}
A point ($\bar{s},\bar{x}$) is a maximal solution of NBI with $\beta \in R^{m}$ , $sum_{i=1}^{m}\beta_{i}=1$, if and only if $(-\bar{s},\bar{x})$ is a minimal solution of $\overline{SP}(a,r)$ with $a=f^{*}+\Phi\beta$ and $r=-\bar{n}$
\end{NBISPLemma}

\subsection{NBI and MSF}
Another connection is between NBI and modified shortage function as is proved in the next proposition:
\renewcommand{\thesubsection}{\arabic{subsection}}
\newtheorem{NBIMSFtheorem}{Proposition}[subsection]
\begin{NBIMSFtheorem}{}\label{NBIMSFtheorem}
Modifed shortage function scalarization is equivalent to NBI by the following substitutions:
\begin{equation}
\begin{split}
s&=\delta  \\
f_{1}(x)&=E[R(x)]  \\
f_{2}(x)&=V[R(x)]  \\
f_{3}(x)&=Sk[R(x)] \\
\Phi \beta+f^{*}&=c  \\
\bar{n}&=g
\end{split}
\end{equation} 
where $c=\begin{bmatrix}
           E[R(y^{k})] \\
           Var[R(y^{k})] \\
           Sk[R(y^{k})]
         \end{bmatrix}$
\begin{proof}
proof is straightforward by a simple substitution.
\end{proof} 
\end{NBIMSFtheorem}

\subsection{SP and SF}

\renewcommand{\thesubsection}{\arabic{subsection}}
\newtheorem{equivalencetheorem}{Proposition}[subsection]
\begin{equivalencetheorem}{}\label{equivalence}
Shortage function scalarization is equivalent to pascoletti serafini scalarization by the following substitutions:
\begin{equation}
\begin{split}
\delta&=-t  \\
a_{1}&=E(R(y^{k'}))=2E(R(x))-E(R(y^{k}))  \\
a_{2}&=V(R(y^{k'}))=2V(R(x))-V(R(y^{k}))  \\
a_{3}&=E(R(y^{k'}))=2S(R(x))-S(R(y^{k}))  \\
r_{1}&=g_{E}  \\
r_{2}&=g_{V}  \\
r_{3}&=g_{S}
\end{split}
\end{equation} 
\begin{proof}
Since $max \ -t=min \ t$ and the reference point can be set as any point in $R^3$, substituting the reference point(a) and the direction(r) in SP constaint which is $a+tr-f(x)\in K$ and choosing the cone K to be trivial $R^{3+}$ the proof is complete. 
\end{proof} 
\end{equivalencetheorem}

\subsection{NBI and goal programming}
It is shown here that the popular goal programming(PGP) method which is widely used in portfolio optimization literature has close connection with NBI and is defined as follows:
\renewcommand{\thesubsection}{\arabic{subsection}}
\newtheorem{PGPdefinition}{Definition}{subsection}
\begin{PGPdefinition}{}\label{PGPdefinition}
Polynomial goal programming (PGP) is defined as:
\begin{equation} 
\begin{split}
PGP(\alpha,\beta)&=min \{d_{1}^\alpha+d_{3}^{\beta};d_{1}=z_{1}^{*}-z_{1},z_{2}=1,d_{3}=z_{3}^{*}-z_{3}     \}  \\
z_{1}^{*}&=max \{z_{1};z_{2}=1 \} \\
z_{3}^{*}&=max \{z_{3};z_{2}=1 \}
\end{split}
\end{equation}
\end{PGPdefinition}
\renewcommand{\thesubsection}{\arabic{subsection}}
\newtheorem{NBIPGPtheorem}{Proposition}[subsection]
\begin{NBIPGPtheorem}{}\label{NBIPGPtheorem}
A solution to NBI problem is also a solution to PGP portfolio optimization problem.
\end{NBIPGPtheorem} 
\begin{proof}
Since $(x^{*},s^{*},\lambda^{*})$ is the solution of NBI it satisfies the first order KKT conditions:
\begin{equation} \label{eq-stationarityForNBI}
\begin{split}
\nabla_{X}F(x^{*})\lambda^{*}&=0  \\
1+\hat{n}\lambda^{*}&=0
\end{split}
\end{equation}
where $\lambda\in R^{3}$ represents the multipliers corresponding to the 3 equality constraints namely return,variance and skewness constraint.
On the other hand, the first oder KKT condition for PGP can be decomposed for two different sets of coordinates. Stationarity equations for the first set of coordinates $(\omega_{1},\omega_{2},\omega_{3})$ results:
\begin{equation}  \label{eq-stationarityForFirstSet} 
\nabla_{\omega}Z \ \mu^{\star}=0
\end{equation} 
Now stationarity with respect to second set of coordinates $(d_{1},d_{3})$ yields:
\begin{equation}  \label{eq-stationarityForSecondSet} 
\begin{split}
\alpha d_{1}^{\alpha -1}+\mu^{(1)}&=0  \\
\mu^{(2)}&=0  \\
\beta d_{3}^{\beta -1}+\mu^{(3)}&=0  
\end{split}
\end{equation}
Now expanding $1+\hat{n}\lambda^{star}=0$ in \ref{eq-stationarityForNBI} generates:
\begin{equation}  \label{expandedInnerProduct}
1+\hat{n}_{1}\lambda_{1}^{*}+\hat{n}_{2}\lambda_{2}^{*}+\hat{n}_{3}\lambda_{3}^{*}=0
\end{equation}
Simplifying \ref{eq-stationarityForFirstSet} would produce:
\begin{equation}   \label{simplified}
1+\alpha \frac{d_{1}^{\alpha -1}}{\mu^{(1)}}=0   
\end{equation}
combining \eqref{simplified} and \eqref{expandedInnerProduct} and solving for 
$\alpha$ , $\beta$ results:
\begin{equation}  \label{alphabeta}
\begin{split}
\alpha&=\frac{\hat{n}_{1}\lambda_{1}^{*}+\hat{n}_{2}\lambda_{2}^{*}+\hat{n}_{3}  \lambda_{3}^{*}}{d_{1}^{\alpha-1}}   \\
\beta&=-\frac{\mu{(3)}}{d_{3}^{\beta-1}}
\end{split}
\end{equation}
Equivalently, $(\omega^{*},d^{*},\mu^{*})$ is the solution of PGP problem and since F and X and in NBI problem are the same as Z and $\omega$ in PGP problem and for any solution and lagrange multipliers there exists equivalent ones by suitable substitutions for $\alpha$ and $\beta$ like \eqref{alphabeta} the proof is complete.   
\end{proof}

\section{Methodology}
There are some quality measures that can be used to figure out which algorithm is better than the other. These measures are well explained in \citep{Eichfelder2008} namely coverage error($\epsilon$), uniformity level($\delta$) and cardinality(number of points) \\
These three measures are conflicting in nature and another multiobjective optimization on top of the main optimization problem may be formed for a rigorous analysis. Satisfying these type of measures are what most researchers refer to as adative considerations. In another context, there are two important aspects in any multiobjective optimzation namely accuracy and diversity. The former forces the solutions to converge to pareto frontier while the latter makes the efficient set equidistanced as much as possible. There are many other measures in the literature such as hypervolume indicator but implementing some of them makes the algorithm very slow and even convergence of them are not proved.
\begin{figure}
  \includegraphics[width=\linewidth]{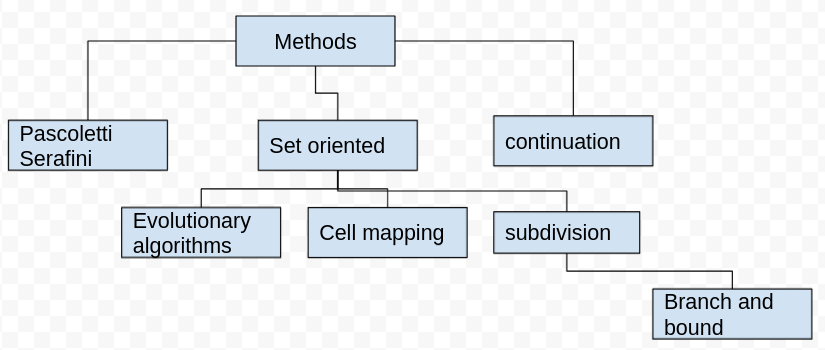}
  \caption{multiobjective optimization methods}
  \label{fig:methods}
\end{figure} 
\\
So far, mostly Pascoletti Serafini methods are discussed in the present paper but there are other ideas in the literature as depicted in Figure~\ref{fig:methods}. Set oriented methods steer a set of solutions at each iterations such as \citep{Hernandez2013} for cell mapping method or \citep{Dellnitz2005} for subdivision algorithms.
In the present paper, two adaptive algorithm for higher moment multiobjective portfolio optimization are given. The first one is an adaptive epsilon constraint method while the second one is not based on scalarization and has its root in \citep{Hillermeier2001} but is recently developed by many researchers as in \citep{Martin2018} , \citep{Schutze2020} and also generalized for problems having inequality constraints in \citep{Beltran2020}. 
Both of the proposed algorithms in this section are based on KKT conditions but the approaches are slightly different. Both methods are designed to produce equidistance pareto frontier points. The equidistance parameter in the first algorithm is $\alpha$ while in the second algorithm it is called $\tau$ to mimick the variables in the related historical articles. So both methods are adaptive in the sense of equal distance points on efficient frontier but no other considerations are taken for the delta and cardinality quality measures since they are expected to produce good results, otherwise they add to the complexity of the algorithms.The first algorithm is based on Epsilon constraint which is a special case of Pascoletti Serafini while the second algorithm is based on continuation methods and the connections are well shown in Figure~\ref{fig:methods}.

\subsection{Adaptive Epsilon Constraint}
Adaptive ECS is described as 
\begin{equation}  \label{PKepsilon}
\begin{split}
min &\ f_{k}(x)  \\
f_{i}(x) &\leq \epsilon_{i} , i \in \{1,\hdots,m\} \setminus \{k\}   \\
x \in \Omega   
\end{split}
\end{equation}

The scalar optimization problem \eqref{PKepsilon} can be formulated as \\
\begin{equation}  \label{eq-PKepsilon}
\begin{split}
min \ t \\  
\epsilon_{i}-f_{i}(x) &\geq 0  \  i\in \{1,\hdots,m\} \setminus \{k\} \\
t-f_{k}(x)&\geq 0  \\
g_j(x)&\geq 0 \  j\in \{1,\hdots,p\} \\
h_{l}(x)&=0   \  l\in \{1,\hdots,q\}  \\
t \in R &, \  x \in R^{n}  \\
\end{split}
\end{equation}
It is proved by Theorem 2.27 in \citep{Eichfelder2008} how it is possible to relate SP to epsilon-constraint method via lagrange multipliers. So using the following substitutions for $a_i$ and r, if $\bar{x}$ is a minimal solution of  \eqref{eq-PKepsilon} . Thus, \eqref{eq-PKepsilon} is equivalent to SP(a,r) with the following substitutions for $a_i$ and r :
\begin{equation}
\begin{split}
a_{i}&=\epsilon_i \ \forall i \in \{1,\hdots,m\} \setminus \{k\}   \\
a_{k}&=0  \\
r&=e_{k}
\end{split}
\end{equation}  
The full algorithm is shown in algorithm~\ref{alg:algorithm-label}.
\begin{algorithm}[H]
\caption{Adaptive $\epsilon$ constraint method for mean-variance-skewness }
\label{alg:algorithm-label}
Input: Choose the desired number $N_1$ of discretization points
       for the range of the functions $f_1$(i.e. in direction $v^1=(1,0,0)^T$ and $N_2$ for the range of function $f_2$ (i.e. in direction  $v^2=(0,1,0)^T$ )  \\
       Step 1: solve the optimization problems $min_{x \in \Omega} f_{i}x$
       with minimal solution $x^{min,i}$ and minimal value $f_i(x^{min,i})$
       for i=1,2 as well as $max_{x \in \Omega} f_{i}x$ with maximal solution $x^{max,i}$ and maximal value $f_{i}(x^{max,i})=\epsilon_{i}^{max}$ for i=1,2  \\
       Step 2: Set  $L_i=\frac{\epsilon_{i}^{max}-\epsilon_{i}^{min}}{N_i}$
       and solve the problem $P_3(\epsilon)$ \\
       for all parameters $\epsilon \in E$ with 
       $E:=\{ \epsilon=(\epsilon_1,\epsilon_2)\in R^2 | \epsilon_i=\epsilon_{i}^{min}+\frac{L_{i}}{2}+l_i.L_i$  \\
       for $l_i=0,...,N_{i-1} , i=1,2  \}$  \\
       Determine the set 
       $A^{E}=\{ (\epsilon,\bar{x},\bar{\mu}) | \bar{x} \} $ where $\bar{x}$ is a minimal solution of 
       $(P_{3}(\epsilon))$ with parameter $\epsilon$  and lagrange multiplier  $\bar{\mu}$  
        to the constraints $f_i(x) \leq  \epsilon_{i}$ \\
        $,i=1,2$  for $\epsilon \in E   $ \\
        Step 3: Determine the set
        $D^{H^0,f}:=\{f(x)|\exists \epsilon \in R^{2}, \mu \in R^{2}_{+} with (\epsilon,x,\mu)\in A^{E} \} $  \\
        Input: Choose $y\in D^{H^{0},f}$ \ with $y=f(x^{\epsilon})$ and $(\epsilon ,x^{\epsilon},\mu^{\epsilon}) \in A^{E}$
        if y is a sufficient good solution, then stop. Else, if additional points in the neighborhood of y are desired, give a distance $\alpha \in R$, $\alpha>0$ in the image space and the number of desired new points $\bar{n}=(2k+1)^2-1$ (for a $k \in N$) and go to Step 4. \\
        Step 4: Set 
        $\epsilon^{ij}:=\epsilon+i.\frac{\alpha}{1+(\mu_{1}^{\epsilon})^{2}}  \begin{bmatrix}
           1 \\
           0 \\
         \end{bmatrix} + j.\frac{\alpha}{1+(\mu_{2}^{\epsilon})^{2}}  \begin{bmatrix}
           0 \\
           1 \\
         \end{bmatrix} $   \\
         for all  \\
         $(i,j) \in \{ (i,j)\in Z^{2} | i,j \in \{-k,\hdots,k\},(i,j) \neq (0,0)     \} $
                
and solve problem $(P_{3}(\epsilon^{i,j}))$  \\
if there exists a solution $x^{i,j}$ with Lagrange multiplier $\mu^{i,j}$, then set $A^{E}:=A^{E} \cup \{ \epsilon^{i,j}, x^{i,j} , \mu^{i,j}  \}$. \\
Go to step 3. \\
Output: The set $D^{H_{0},f}$ is an approximation of the set of weakly efficient points 
\end{algorithm}
\begin{figure}[H]
  \includegraphics[width=\linewidth]{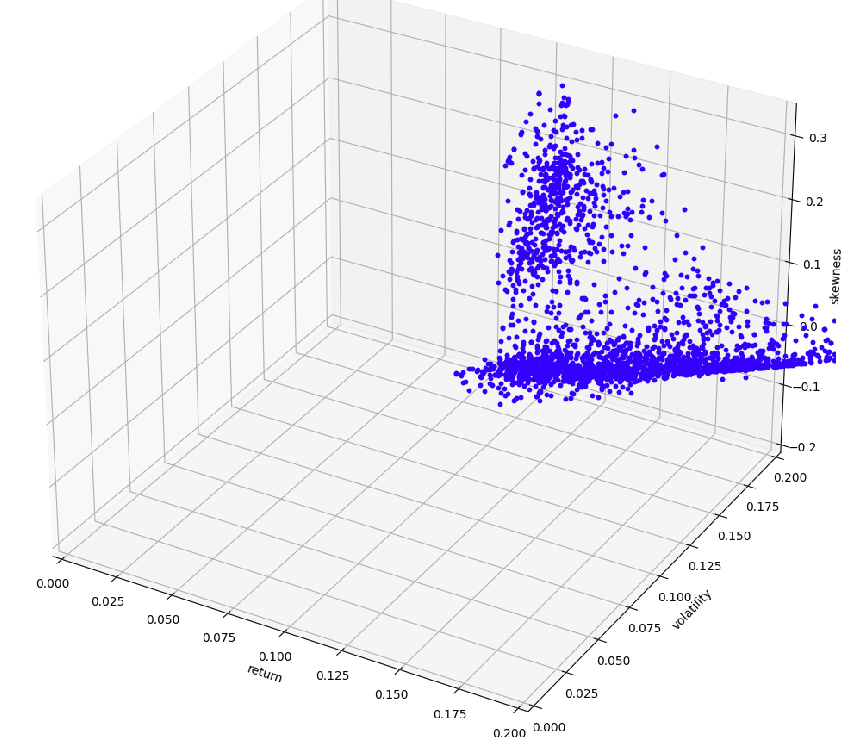}
  \caption{Pareto Front using adaptive epsilon constraint method with 2500 points}
  \label{fig:EC2zoom}
\end{figure}
The simulation results for adapative epsilon constraint method is illustrated in Figure~\ref{fig:EC2zoom}. As is explained in the previous section $\epsilon$ constraint method is just a special case of SP and the following algorithm is used for multiobjective portfolio optimization having three objectives namely return, variance and skewness.
\subsection{Adaptive Multi-start Pareto Tracer}
There are two main ideas in this approach: 1-KKT condition in single objective is generalized to multiobjective as is described in \citep{Hillermeier2001}. 2-A predictor corrector idea which first predicts the next move in decision space and then corrects it by a multiobjective gradient descent. The algorithm \ref{alg:AMPT} is a modification of \citep{Martin2018} by doing it in a multi start way and shaping the objective space distribution and customizing it for portfolio optimization of three objectives namely mean, variance and skewness.
Consider the multiobjective optimization problem defined bellow:
\begin{equation}
\begin{split} \label{eq-PT-MOP}
min \  &F(x) \\
s.t \ h(x)&=0
\end{split}
\end{equation}
where F is a vector of objectives $F:R^{n}\rightarrow R^{k}$. F in \ref{eq-PT-MOP} is actually a three dimensional vector of mean,variance and skewness and decision space has dimension n which refers to number of assets or factors in a multifactor investment framework. h in \ref{eq-PT-MOP} is the constraint that the sum of allocations to different assets should be one. A predictor corrector method is developed in \citep{Hillermeier2001} by considering
\begin{equation} \label{eq-Hiler}
\tilde{F}(x,\alpha)=\begin{bmatrix}
         \sum_{i=1}^{k}\alpha_{i}\nabla f_{i}(x) \\
         \sum_{i=1}^{k}\alpha_{i}-1  
         \end{bmatrix}=0
\end{equation}
The set of KKT points of \ref{eq-PT-MOP} is contained in the null set of $\tilde{F}$ which is the idea behind many continuation methods along $\tilde{F}^{-1}(0)$  as written in \ref{eq-Hiler}. So a simple representation for the tangent vectors to Pareto set can be written as
\begin{equation} \label{eq-Hiler-prime}
\tilde{F}(x,\alpha)\begin{bmatrix}
\nu  \\
\mu
\end{bmatrix}
=\begin{bmatrix}
         \sum_{i=1}^{k}\alpha_{i}\nabla^{2} f_{i}(x) & \nabla f_{1}(x) & \hdots & \nabla f_{k}(x)  \\
         \sum_{i=1}^{k}\alpha_{i}-1    & 1 & \hdots & 1
         \end{bmatrix} \begin{bmatrix}
         \nu  \\
         \mu
         \end{bmatrix}=\begin{bmatrix}
         0  \\
         0
         \end{bmatrix}
\end{equation}
Thus the vector in \ref{eq-Hiler-prime} can be expressed as
\begin{equation} \label{eq-vmu}
v_{\mu}=-W^{-1}_{\alpha}J^{T}\mu 
\end{equation}
where $W_\alpha$ in \eqref{eq-vmu} is defined as
\begin{equation}
W_\alpha:=\sum_{i=1}^{k}\alpha_{i}\nabla^{2}f_{i}(x) \in R^{n\times m} 
\end{equation}
and J is defined by
\begin{equation}
J=J(x)=\begin{bmatrix}
\nabla f_{1}(x)^{T} \\
\vdots  \\
\nabla f_{k}(x)^{T}  \in R^{k \times n}
\end{bmatrix}
\end{equation}
Now d has a special meaning that expresses the first order approximated movement in objective space for infinitesimal step sizes and is defined below
\begin{equation}  \label{eq-d}
d:=J\nu
\end{equation}
Using \eqref{eq-vmu} and the definition of d in \eqref{eq-d} makes
\begin{equation}  \label{eq-jmuIsD}
J\nu_{\mu_{d}}=-JW^{-1}_{\alpha}J^{T}\mu=d
\end{equation} 
Now any possibility for selecting d will generate a different distribution on pareto front. Since the resulting $\nu_{mu}$ are tangents to Pareto set and the aim is making even spread of points on Pareto front, directions d should be selected such that it makes orthonormal basis of tangent space to Pareto front at F(x). One of the natural ways to do this is utilizing QR factorization of $JW^{-1}J^{T}$ and by selecting 
\begin{equation}
d_{i}:=q_{i+1}  \  i=1,\hdots,k-1
\end{equation}
$\mu$ can be solved. Now it is possible to obtain the predictor $p:=x+t\nu$ to make an evenly distributed set of solutions along the Pareto front the following approximation for t can be chosen
\begin{equation}
\begin{split}
\norm{F(x_{i})-F(x_{i+1}) &\approx \tau}  \\
t=\frac{\tau}{\norm{J\nu_{\mu}}}
\end{split}
\end{equation} 
The predictor part explained so far cares about going through the tangent direction in Pareto set to create an equidistanced set of Pareto front points. The next part is the corrector which cares about convergence issues and is explained in \citep{Fliege2009}and the present paper implements it for the corrector part of the algorithm, but there are other methods that can be used for this step such as \citep{Povalej2014}   which approximates the second derivative matrices instead of evaluating them although both methods have superlinear rate of convergence. The corrector part implemented in the present paper is based in minimization of the following optimization problem.
\begin{equation}  \label{eq-corrector}
\begin{split}
min \  g(t,s)&=t  \\
s.t \ \nabla F_{j}(x)^{T}s+\frac{1}{2}s^{T}\nabla^{2}F_{j}(x)s-t&\leq 0   \  (t,s)\in R \times R^{n}  
\end{split}
\end{equation}
Thus, the full predictor-corrector algorithm is shown in Algorithm~\ref{alg:AMPT}.
\begin{algorithm}[H]
\caption{Adaptive Multi-start Pareto Tracer for mean-variance-skewness }
\label{alg:AMPT}
Repeat Predictor-corrector loop: \\
Input : create some bundles to prepare for the multistart algorithm
Predictor part:  \\
Step 1: calculate $\mu$ by solving \eqref{eq-jmuIsD} \\
Step 2: calculate direction $v_{mu}$ from \eqref{eq-vmu}  \\
Step 3: update the predictor position in decision space by $p:=x+t\nu_{\mu}$ \\
Corrector part:  \\
Step 4: solving \eqref{eq-corrector} and update x  
\end{algorithm}    
\begin{figure}
  \includegraphics[width=\linewidth]{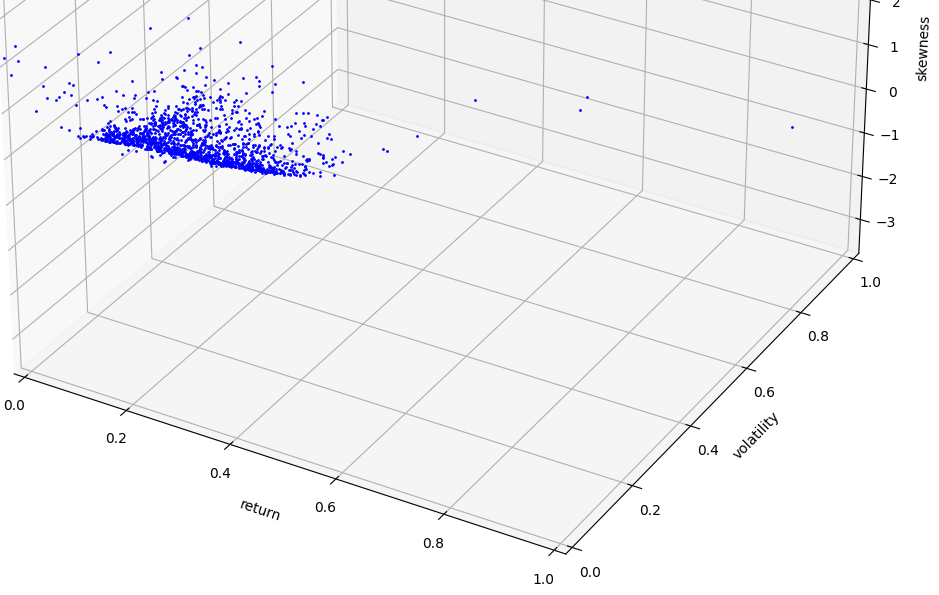}
  \caption{adaptive multistart Pareto Tracer}
  \label{fig:PT1}
\end{figure}
Simulation results are shown in  Figure~\ref{fig:PT1}. Since the current algorithm has many parameters to tune and to get all parts of pareto front in a faster way, an evolutionary algorithm could be combined with the current algorithm to make a hybrid algorithm.
\section{Conclusions}
The paradigms in higher moment multiobjective optimization are critically reviewed and the connection between some of them are explained in the present paper. It has been proved that shortage function method can be seen as a Pascoletti Serafini scalarization. Finally, two algorithms for portfolio optimization are suggested. The first one is based on scalarization paradigm and is called adaptive epsilon constraint method while the second one is a type of continuation method and is called adaptive multistart Pareto Tracer which bundles different local solutions to provide a global Pareto Front by both exploration and exploitation.
 
\section{Future Works}
The first suggested algorithm can be modified by handling variable structure ordering as is explained in \citep{Eichfelder2014} and \citep{Eichfelder2012}. So instead of a fixed cone K for ordering, a variable cone is considered and each point could have a different ordering corresponding to a different cone.
The second suggested algorithm can be further developed by hybridizing it with evolutionary algorithms such as genetic algorithm to make it faster.
Another line of research which could be theoretically interesting is to see if the continuation method suggested in the second algorithm could be seen as generalization of Pascoletti Serafini ,since the continuation framework is directly working on KKT conditions for a multiobjective optimization problem while scalarization takes advantage of KKT condition of a mono objective problem.
\bibliographystyle{agsm}
\bibliography{portfolioOptimization}
\end{document}